\documentclass[floatfix, superscriptaddress, reprint, prl,aps,10pt]{revtex4-1}
\usepackage[utf8]{inputenc}
\usepackage{amsmath}
\usepackage{amsfonts}
\usepackage{amssymb,amsthm}
\usepackage{graphicx}
\usepackage[]{hyperref}
\hypersetup{
  colorlinks = true,
  urlcolor = magenta,
  linkcolor = red,
  citecolor = blue
}

\usepackage[all]{hypcap}
\usepackage{url}
\newcommand{\ket}[1]{|#1\rangle}
\newcommand{\bra}[1]{\langle#1|}

\usepackage{textcomp, complexity}
\usepackage[ruled, linesnumbered, procnumbered]{algorithm2e}
\usepackage{xcolor}

\let\oldnl\nl
\newcommand{\nonl}{\renewcommand{\nl}{\let\nl\oldnl}}

\newtheorem{defn}{Definition}

\newtheorem{thm}{Theorem}
\newtheorem{lemma}[thm]{Lemma}
\newtheorem{corollary}[thm]{Corollary}

\makeatletter
\newenvironment{subtheorem}[1]{%
  \def\subtheoremcounter{#1}%
  \refstepcounter{#1}%
  \protected@edef\theparentnumber{\csname the#1\endcsname}%
  \setcounter{parentnumber}{\value{#1}}%
  \setcounter{#1}{0}%
  \expandafter\def\csname the#1\endcsname{\theparentnumber.\Alph{#1}}%
  \ignorespaces
}{%
  \setcounter{\subtheoremcounter}{\value{parentnumber}}%
  \ignorespacesafterend
}
\def\@caption@fignum@sep{\ (Color online).\ }
\makeatother
\newcounter{parentnumber}

\newcounter{proof}

\usepackage[normalem]{ulem}

\begin{abstract}
We make the case for studying the complexity of approximately simulating (sampling) quantum systems for reasons beyond that of quantum computational supremacy, such as diagnosing phase transitions.
We consider the sampling complexity as a function of time $t$ due to evolution generated by spatially local quadratic bosonic Hamiltonians.
We obtain an upper bound on the scaling of $t$ with the number of bosons $n$ for which approximate sampling is classically efficient.
We also obtain a lower bound on the scaling of $t$ with $n$ for which any instance of the boson sampling problem reduces to this problem and hence implies that the problem is hard, assuming the conjectures of Aaronson and Arkhipov [{Proc.\ 43rd {{Annu}}.\ {{ACM Symp}}.\ {{Theory Comput}}.\ {{STOC}} '11}].
This establishes a dynamical phase transition in sampling complexity.
Further, we show that systems in the Anderson-localized phase are always easy to sample from at arbitrarily long times.
We view these results in the light of classifying phases of physical systems based on parameters in the Hamiltonian.
In doing so, we combine ideas from mathematical physics and computational complexity to gain insight into the behavior of condensed matter, atomic, molecular and optical systems.
\end{abstract}

\begin{document}
\title{Dynamical phase transitions in sampling complexity}

\author{Abhinav Deshpande}
\affiliation{Joint Center for Quantum Information and Computer Science, NIST/University of Maryland, College Park, MD 20742, USA}
\affiliation{Joint Quantum Institute, NIST/University of Maryland, College Park, MD 20742, USA}
\author{Bill Fefferman}
\affiliation{Joint Center for Quantum Information and Computer Science, NIST/University of Maryland, College Park, MD 20742, USA}
\affiliation{Electrical Engineering and Computer Sciences, University of California, Berkeley, CA 
94720, USA}
\author{Minh C.\ Tran}
\affiliation{Joint Center for Quantum Information and Computer Science, NIST/University of Maryland, College Park, MD 20742, USA}
\affiliation{Joint Quantum Institute, NIST/University of Maryland, College Park, MD 20742, USA}
\author{Michael Foss-Feig}
\affiliation{United States Army Research Laboratory, Adelphi, MD 20783, USA}
\affiliation{Joint Center for Quantum Information and Computer Science, NIST/University of Maryland, College Park, MD 20742, USA}
\affiliation{Joint Quantum Institute, NIST/University of Maryland, College Park, MD 20742, USA}
\author{Alexey V.\ Gorshkov}
\affiliation{Joint Center for Quantum Information and Computer Science, NIST/University of Maryland, College Park, MD 20742, USA}
\affiliation{Joint Quantum Institute, NIST/University of Maryland, College Park, MD 20742, USA}
\maketitle

In the quest towards building scalable and fault-tolerant quantum computers, demonstration of a quantum speedup over the best possible classical computers is an important milestone and is termed \emph{quantum computational supremacy} \cite{Preskill2012a}.
There are several candidates for tasks where such a speedup could be demonstrated \cite{Terhal2002, Terhal2002a, Bremner2010, Aaronson2011, Jozsa2013, Fefferman2016, Bremner2016, Farhi2016a, Aaronson2016a, Fefferman2017, Gao2017, Harrow2017}, where the problem is to simulate a quantum system in the sense of approximate sampling.
However, there has also been some debate about the required system size before one can claim quantum computational supremacy, due to improved simulation techniques and algorithms \cite{Neville2017, Clifford2018, Dalzell2018}.
In this Letter, we consider the impact of the field of quantum computational supremacy on other areas of physics and show that studying the complexity of simulating physical systems is useful for understanding phase transitions.

Here we consider the classical complexity of approximate sampling, referred to as ``sampling complexity''.
This is the task of producing samples from a distribution close to the probability distribution occurring in a quantum system upon measurement in a standard basis.
This task is a good notion of what it means to simulate physics on a classical computer since it captures how well a computer can mimic an experiment in which one can measure the output at several sites.
When we consider sampling complexity as a function of system parameters, the system can be classified as easy in some regimes and hard in some others.
Since the designations ``easy'' and ``hard'' are exhaustive and there is no smooth way to go from one regime to another, we posit that this transition from easy to hard happens abruptly, a phenomenon very reminiscent of phase transitions.
Just like order parameters are zero on one side of a phase transition and nonzero on another, sampling complexity is different on either side of the transition, and can be used to draw phase boundaries as a function of time and other system parameters.
These boundaries can be different from those drawn by more conventional order parameters, signifying new physics in otherwise well-studied systems.
Indeed, phase transitions in average-case complexity have been studied both in the classical \cite{Kirkpatrick1994} and quantum regime \cite{Laumann2010}.

In this Letter, we show that transitions in sampling complexity \cite{Seshadreesan2015} are indicative of physical transitions.
We consider a system of $n$ bosons hopping from one site to another on a lattice of $m$ sites and study the sampling complexity as a function of time for an initial product state.
We show that it goes from easy to hard as the scaling of evolution time $t$ with the number of bosons $n$ increases, thereby exhibiting a dynamical phase transition \cite{Heyl2013, Heyl2017}.
We find that the timescale at which the complexity changes is the timescale when interference effects start becoming relevant.
We conjecture that in general, this is linked to the Ehrenfest timescale at which quantum effects in a system become considerable \cite{Rozenbaum2017}.
We also show that systems in the Anderson-localized phase are always easy to simulate.
We use the Lieb-Robinson bound \cite{Lieb1972} as an ingredient in our proof of easiness.

\textit{Setup.---} The model consists of free bosons hopping on a lattice in $d$ dimensions (denoted $d$-D) with sites labeled by indices $i,j$.
Our results in this paper can be applied to linear optics as well, with the bosonic sites being replaced by photonic modes.
The Hamiltonian is given by $H = \sum_{i, j=1}^m J_{ij}(t) a_i^\dagger a_j$, where $a_i^\dagger$ is the creation operator of a boson at the $i$'th site.
$J(t)$, which can be time-dependent in general, is an $m \times m$ Hermitian matrix that encodes the connectivity of the lattice.
One way to show hardness in the boson sampling proposal by Aaronson and Arkhipov (AA) \cite{Aaronson2011} is to generate any linear optical unitary $U$ acting on the bosonic sites, in particular, Haar-random unitaries.
Any $U$ can be generated through a free boson Hamiltonian by taking $H = i \log U$ and evolving it for unit time.
However, this Hamiltonian can require arbitrarily long-range hops on the lattice in general.
Since such long-ranged Hamiltonians may not be realistic, we consider the sampling complexity of a bosonic Hamiltonian with nearest-neighbor hops.
$J_{ij}$ is nonzero only if $i = j$ or $i$ and $j$ label adjacent sites.
We further restrict $J_{ij}$ to satisfy $|J_{ij}| \leq 1$ in order to set an energy scale.
Our proofs remain valid even if we allow the diagonal terms $J_{ii}$ to be unbounded.

One can efficiently solve the equations of motion $i \dot{a_i}^\dagger(t) = [a_i^\dagger(t), H(t)]$ on a classical computer to obtain $a_i^\dagger (t) = \sum_k a_k^\dagger (0) R_{ki}(t)$ for some transformation matrix $R$.
From here onward, we shall take $R(t)$ to be the input to the problem, since it can be determined from the input Hamiltonian $H$ and time $t$ in time $\mathrm{poly}(m, \log t)$.

\begin{figure}[]
\includegraphics[width=0.9\linewidth]{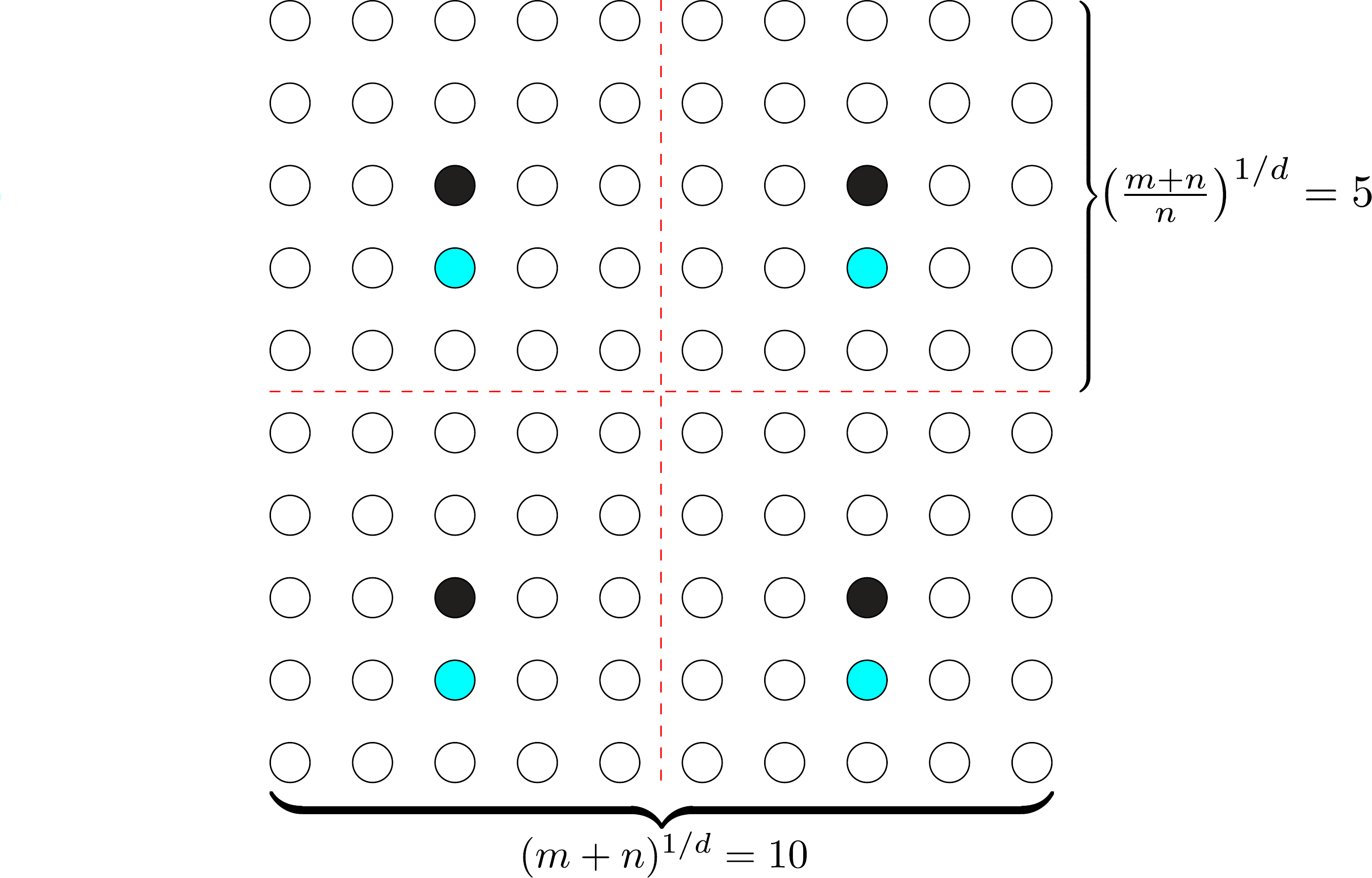}
\caption{An example of the initial state in $d = 2$ dimensions.
Here $m = 96$, $n = 4$, $\beta = 3$ and $c_1 = 3/2$.
The black circles represent sites with a single boson.
The cyan circles represent the ancillas.}
\label{fig_lattice}
\end{figure}

The $m$ sites in the problem are numbered from 1 to $m$, and together with $n$ ancilla sites, are arranged in a lattice of side length $(m+n)^{1/d}$ in $d$ dimensions.
The initial state has $n$ bosons equally spaced in the lattice as shown in Fig.\ \ref{fig_lattice}.
We take $m = c_1 n^\beta$, where $\beta$ controls the sparsity of occupied sites in the lattice and can be set to 5 as required for the hardness of boson sampling \cite{Aaronson2011}.
The minimum spacing between any two bosons in the initial state is $2L = \left(\frac{m + n}{n} \right)^{1/d} > c_1^{1/d}n^\frac{\beta -1}{d}$.
The quantity $L$ is an important length scale in the problem.
The ancillas in the lattice, marked in cyan, are not counted as part of the $m$ sites and are present in order to accelerate the time required to construct an arbitrary unitary, which is useful for the hardness result.
Their presence does not change the scaling of quantities like $L$ with $n$.

The input states are described by vectors of the form $r =(r_1, \ldots, r_m)$, specifying the number of bosons on each site, so that $r_1 + \ldots r_m = n$.
Measurement in the boson number basis defines a distribution $\mathcal{D}_U$, which we aim to sample from.
The probability of finding an output state $s = (s_1,s_2,\ldots,s_m)$ is given by 
\begin{align}
\Pr_{\mathcal{D}_U}[s] = \frac{1}{r! s!} |\mathrm{Per}(A)|^2, \label{eq_probability}
\end{align}
where $r! := r_1! \ldots r_m!$ (with $s!$ defined similarly), $A_{n \times n}$ is a matrix formed by taking $s_i$ copies of the $i$'th column and $r_j$ copies of the $j$'th row of $R$ in any order, and $\mathrm{Per} (A)$ denotes the permanent of $A$ (see Refs.\ \cite{Aaronson2011a, SM_complexity_OP} for details).

For the particular choice of initial states described in Fig.\ \ref{fig_lattice}, the task is to sample from a distribution that is close to $\mathcal{D}_U$ in variation distance when given a description of the unitary $R(t)$.
We now formalize the notion of efficient sampling.

\begin{defn}{Efficient sampler:}
An efficient sampler is a classical randomized algorithm that takes as input the unitary $R_{ij}$ and outputs a sample s from a distribution $\mathcal{D}_{\mathcal{O}}$ such that the variation distance between the distributions $\epsilon = \| \mathcal{D}_\mathcal{O} - \mathcal{D}_U \| \leq O(\frac{1}{\mathrm{poly}(n)})$, in runtime $\mathrm{poly}(n)$ (see notes \footnote{Our easiness results also apply to a stronger notion of approximate sampling, namely that the algorithm samples from an $\epsilon$-close distribution in runtime $\mathrm{poly}(n,1/\epsilon)$}, \footnote{1. $f(n)=O(g(n))$ if $\forall n \geq n_0, f(n) \leq cg(n)$. \\ 2. $f(n) = \Omega (g(n)) \Leftrightarrow g(n) = O(f(n))$. \\ 3. $f(n) = \Theta (g(n)) $ if $f(n) = \Omega (g(n))$ and $f(n) = O(g(n))$.}).
\end{defn}

We call the sampling problem \emph{easy} if there exists an efficient sampler for the problem in the stated regime.
Conversely, the problem is hard if there cannot be an efficient sampler.
Since a negative statement such as the inexistence of an algorithm is difficult to prove, our practical definition of hardness of a sampling problem is if it is at least as hard as boson sampling.
This enables us to use the results of AA \cite{Aaronson2011} to claim the hardness of sampling in some regimes.
In doing so, our hardness results ultimately rely on the truth of AA's conjectures.
One of these conjectures concerns the hardness of additively approximating $|\mathrm{Per}(G)|^2$ for Gaussian-random $G$ with high probability (for which AA give reasonable evidence).
The second, more widely believed conjecture is that the polynomial hierarchy, an infinite tower of complexity classes, does not ``collapse'', i.e. is truly infinite.

We restrict our attention to two special cases where we can show the existence of an efficient sampling algorithm: i) when the system evolves for a time smaller than the system timescale $L/v$ (where $v$ is the Lieb-Robinson velocity of information spreading in the lattice, defined more precisely in Eq.\ (\ref{eq_LRbound})), and ii) when there is Anderson localization in the system \cite{Anderson1958}.
These two cases correspond to a promise on the input unitary $R$.
We now state our main results.

\begin{subtheorem}{thm}
\begin{thm}[Easiness of simulation at short times] \label{thm_mainthm} 
For $\beta > 1$ and for all dimensions d, the sampling problem is easy for all $t \leq 0.9L/v$, i.e.\ $\forall t \leq c_2 n^{(\beta - 1)/d}$ for some constant $c_2$.
\end{thm}
The intuition behind this Theorem is that when the time is smaller than the Lieb-Robinson timescale of particle interference $L/v$, the dynamics is approximately classical (in the sense that the particles are distinguishable).

An important technical achievement of this paper is to prove rigorously that this intuition is correct.
This is done in Lemma \ref{lem_main} by showing that the approximation works.
\begin{thm}[Hardness of simulation at longer times] \label{thm_hardness} (based on Theorem 3 of Ref.\ \cite{Aaronson2011})
When $t = \Omega(n^{1+\beta/d})$), the sampling problem is hard in general.
\end{thm}
\end{subtheorem}
For this result, we show that we can apply any unitary $R$ after the stated time.
Therefore, if an efficient sampler exists for this problem, then we have an efficient boson sampler, which cannot exist by our assumption.

The result for the case of Anderson localization comes out as a corollary from Theorem \ref{thm_mainthm}:
\begin{corollary}[Easiness of Anderson-localized systems]
\label{thm_Anderson} For Anderson-localized systems in any dimension d, the sampling problem is easy for all times.
\end{corollary}
The easiness of sampling for Anderson-localized systems is analogous to results showing efficient simulation of localized systems according to various definitions \cite{Friesdorf2015, Huang2015, Pollmann2016, Yu2017, Chapman2017}.

\textit{Easiness at short times.---} In this section, we prove Theorem \ref{thm_mainthm} and Corollary \ref{thm_Anderson}.
First, let us examine the promise we have on the unitary $R$ in both cases.
We use the Lieb-Robinson bound \cite{Lieb1972} on the speed of information propagation in a system.
Applying the bound to our Hamiltonian, we get
\begin{gather} \label{eq_LRbound}
 |[a_i(t), a_j^\dagger(0) ] | = |R_{ij}(t)| \leq \mathrm{min} \left(1, \exp \left( \frac{vt - \ell_{ij}}{\xi}\right) \right),
\end{gather}
where $\ell_{ij}$ is the distance between two sites $i$ and $j$, $v$ is the upper bound to the velocity of information propagation called the Lieb-Robinson velocity and $\xi$ is a length scale).
Note that the results from Ref.\ \cite{Eisert2009} do not apply since we have free bosons here and we work in the single-particle subspace.
The Lieb-Robinson velocity is at most $4(1+2de)$ \cite{Hastings2010} when $|J_{ij}| \leq 1$ and $\xi = 1$.

When the Hamiltonian is Anderson-localized, the unitary $R$ satisfies the following promise at all times \cite{Frohlich1985}:
\begin{gather} \label{eq_localised_eigv}
 |R_{ij}| \leq \exp \left( \frac{-\ell_{ij} }{\xi}\right).
\end{gather}
Here, $\xi$ is the maximum localization length among all eigenvectors.
Eq.\ (\ref{eq_localised_eigv}) can be viewed as a consequence of a Lieb-Robinson bound with zero velocity \cite{Hamza2012}.
On account of the zero-velocity Lieb-Robinson bound, all results for the time-dependent case can be ported to the Anderson localized case, setting $v = 0$.

We give an algorithm that efficiently samples from the output distribution for short times $t <  \frac{0.9 L}{v} = O(n^{(\beta -1)/d})$, given the promise in Eq.\ (\ref{eq_LRbound}).
The algorithm outputs a sample from a distribution $\mathcal{D}_{DP}$, obtained by assuming that the bosons are distinguishable particles, ignoring the effects of interference.
The algorithm is described in more detail in Ref.\ \cite{SM_complexity_OP}.

\textit{Analysis.---}
We prove the correctness of the algorithm by showing that the variation distance between the distributions is upper bounded by an inverse exponential in $n$.
When the bosons are distinguishable, their dynamics is given by a Markov process, described by the matrix $\mathcal{P}_{kl} = |R(t)|_{kl}^2$.
The probability of getting an outcome $s$ is given by
\begin{gather}
 \Pr_{\mathcal{D}_{DP}}[s] = \sum_\sigma \frac{1}{s!}
\mathcal{P}_{\mathsf{in}_1,\mathsf{out}_{\sigma(1)}}
\mathcal{P}_{\mathsf{in}_2,\mathsf{out}_{\sigma(2)}} \ldots
\mathcal{P}_{\mathsf{in}_n,\mathsf{out}_{\sigma(n)}},
\end{gather}
where the sum is over all permutations $\sigma$ mapping the input bosons to the output ones.
In the above equations, $\mathsf{in}_i$ is the site index of the $i$'th boson in nondecreasing order in the input and $\mathsf{out}$ is defined similarly at the output (see Ref.\ \cite{SM_complexity_OP} for an example).
We now state a result on how close $\mathcal{D}_{DP}$ is to the true distribution $\mathcal{D}_U$.

\begin{lemma} \label{lem_main}
When $\beta > 1$ and $t \leq 0.9L/v$, the variation distance satisfies $\|\mathcal{D}_{DP} - \mathcal{D}_U\| = O(\exp[2 \frac{vt-L}{\xi} + 2(d-1) \log L]).$
\end{lemma}
This Lemma makes intuitive sense: because of the Lieb-Robinson bound Eq.\ (\ref{eq_LRbound}) and the fact that the initially occupied sites are separated by a minimum distance $\Theta(L)$ from each other, it takes a time $t = \Theta (L/v)$ for the bosons to start interfering considerably.
Therefore, the classical and quantum distributions agree exponentially closely in $L$ when $t \leq 0.9L/v$.
For a proof of Lemma \ref{lem_main}, see Ref.\ \cite{SM_complexity_OP}.
Assuming this Lemma, we now show Theorem \ref{thm_mainthm}.

\begin{proof}[Proof of Theorem \ref{thm_mainthm}] \refstepcounter{proof} \label{pf_mainthm} Lemma \ref{lem_main} shows that the algorithm samples from a distribution with exponentially small error in $n$, since $L = \Theta (n^{(\beta- 1)/d})$.
To complete the proof of Theorem \ref{thm_mainthm}, we need to show that the runtime of the algorithm is polynomial in $n$.
This is true because the corresponding Markov process of $n$ distinguishable bosons walking on $m$ sites for one step is efficiently simulable: for each of the $n$ particles, we select one among the $m$ sites to walk to, based on the matrix elements of $\mathcal{P}$.
This takes time $O(n\cdot \mathrm{poly}(m))$ = $O(\mathrm{poly}(n))$ to simulate on a classical computer.
\end{proof}

\textit{Hardness at longer times.---}
If we allow the system to evolve for a longer amount of time, we can use the time-dependent control to effect any arbitrary unitary and implement any boson sampling instance in the system.
We can perform phase gates on a site $k$ by setting $J_{kk}$ to be nonzero for a particular time, with the hopping terms and other diagonal terms set to zero.  
We can apply a nontrivial two-site gate between adjacent sites, for example the balanced beamsplitter unitary on the sites 1 and 2, $U = \frac{1}{\sqrt{2}}
\begin{pmatrix}
1 & -1\\
1 & 1
\end{pmatrix}$, by setting $H = -i(a_1^\dagger a_2 - a_1 a_2^\dagger)$ for time $t = \frac{\pi}{4}$.
One can also apply arbitrary unitaries using arbitrary on-site control $J_{ii}(t)$ and fixed, time-independent nearest-neighbor hopping $J_{ij}(t) = 1$.

Using the constructions in Refs.\ \cite{Reck1994, Clements2016}, we can effect any arbitrary $m\times m$ unitary on the $m$ sites with $O(m^2)$ depth from a nontrivial beamsplitter and arbitrary single-site phase gates.
A construction of AA that employs ancillas to obtain the desired final state rather than applying the full unitary on all the sites can be used to achieve a depth $O(nm^{1/d})$.
Each of the $n$ columns of the unitary are implemented in time $O(m^{1/d})$, which corresponds to the timescale set by the Lieb-Robinson velocity and the distance between the furthest two sites in the system (see note \footnote{AA's construction applied each column of the unitary in $O(\log m)$-depth, whereas we can only apply it in $O(m^{1/d})$ depth because of the spatial locality of the Hamiltonian and the Lieb-Robinson bounds that follow.}).

\begin{proof}[Proof of Theorem \ref{thm_hardness}]
From the above, when $t = \Omega(n^{1 +\frac{\beta}{d}})$, we see that we can effect any arbitrary unitary.
This implies that an efficient sampler for this regime can also be used as an efficient boson sampler, which is widely believed not to exist because of AA's results \cite{Aaronson2011}.
\end{proof}

\begin{figure}
\includegraphics[width=\columnwidth]{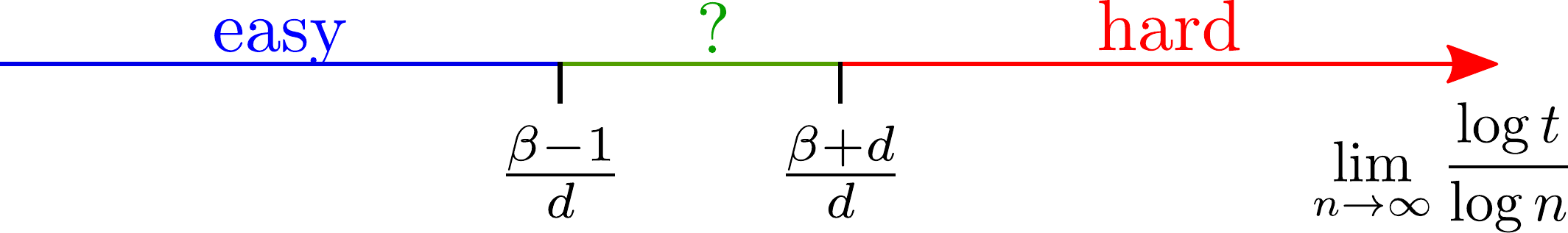}
\caption{Complexity phase diagram of free bosons, illustrating that sampling complexity can delineate boundaries of a physical system as a function of system parameters, including time.
Our results indicate that $\frac{\beta-1}{d} \leq c \leq \frac{\beta+d}{d}$, where $c$ is the transition point for the scaling exponent of $t$ with $n$.}
\label{fig_phasediagram}
\end{figure}

\textit{Outlook.---}
We have defined the sampling problem for local Hamiltonian dynamics and given upper and lower bounds for the scaling of time $t(n)$ with the number of bosons $n$ for which the problem is efficiently simulable or hard to classically simulate, respectively.
Our results are captured in Fig.~\ref{fig_phasediagram} that illustrates the complexity phase diagram of the system.
For time-independent systems, we observe that boson sampling is classically easy for all times if the system is Anderson-localized.
In a future work \cite{Deshpande}, we show that there is a class of static, local Hamiltonians that generate a hard-to-sample output distribution at some time that is not formally infinity.
This means that sampling complexity distinguishes Anderson-localized and delocalized systems, which makes it similar to an order parameter that distinguishes different phases.

The case with nonzero Lieb-Robinson velocity [Eq.\ (\ref{eq_LRbound})] shows two regimes of the scaling of $t$ with $n$ where sampling is provably easy/hard.
We have shown that sampling is easy when $t \leq t_\mathrm{easy} = \Theta(n^{\frac{\beta -1}{d}})$ and hard when $t \geq t_\mathrm{hard} = \Theta(n^{1 + \frac{\beta}{d}})$.
Since our definitions of easiness and hardness are exhaustive, we argue that there must exist a constant $c$ such that sampling is efficient for $t < \Theta(n^c)$ and hard otherwise, illustrating a dynamical phase transition.
Our proof implies that $c \in [\frac{\beta - 1}{d}, \frac{\beta+d}{d}]$ and we show in future work \cite{Deshpande} that the transition is sharp (see note \footnote{The transition here is sharp if there is a specific time $t_{\mathrm{samp}} = \Theta (n^c)$ such that sampling is easy for all $t < t_{\mathrm{samp}}$ and hard otherwise. The transition is coarse if sampling is easy for all $t = \Theta (n^c)$ and hard for all $t$ such that $\lim_{n\rightarrow \infty}\frac{t}{n^c} \rightarrow \infty $.}).
The transition is between two regimes, one for short times in which the system's dynamics is essentially indistinguishable from classical dynamics; and the other in which quantum mechanical effects dominate to such an extent as to forbid an efficient classical simulation.

This result may be viewed as a generalization of a similar result in Ref.\ \cite{Brod2015}, where it was shown that exact boson sampling for depth-4 circuits is hard.
The results there are not directly comparable to ours, since Ref.\ \cite{Brod2015} assumes $\beta = 1$, whereas our results need $\beta > 1$ (easiness) and $\beta  \geq 2\ \mathrm{or}\ 5$ (hardness).
The reason we get easiness even after polynomial time is that we deal with approximate sampling, a less stringent notion of simulation.
In addition, adapting the hardness proof of Ref.\ \cite{Brod2015} into our setup, we obtain that the system is hard to exactly sample from at timescale $\Omega(8L/v) = \Omega(n^\frac{\beta-1}{d})$ \cite{Deshpande}, partially answering what happens in the subregion with the question mark in Fig.\ \ref{fig_phasediagram}.

Recent studies \cite{Rozenbaum2017, Banerjee2017, He2017, Syzranov2017} have also studied transitions based on time-dependent features of the out-of-time-ordered correlator.
This raises the question of the connection between sampling complexity and scrambling time \cite{Hosur2016a,*Roberts2017} in quantum many-body systems \cite{Huang2016, *Fan2017, *Chen2016a, *Swingle2017} and fast scramblers like the Sachdev-Ye-Kitaev \cite{Sachdev1993, Kitaev2015, Maldacena2016, Maldacena2016a, Banerjee2017} models, and black holes \cite{Sekino2008, Lashkari2013, Swingle2016a, Brown2016, *Brown2016a}, where one can explore the connection to recent conjectures on complexity in the dual CFT \cite{Aaronson2016b, Brown2016, *Brown2016a, Brown2018}.
Further, it would be interesting to study sampling complexity in various other physical settings \cite{Marzolino2016}, like systems with topological or many-body-localized phases and systems in their ground state to explore the connection with Hamiltonian complexity.

We therefore propose another motivation for considering sampling complexity beyond the field of quantum computational supremacy: complexity provides a natural way to classify phases of matter that is complementary to traditional approaches based on symmetries and topology.
This is akin to how the study of entanglement in many-body physics has helped us understand phases of matter \cite{Amico2008, Zeng2015} and characterize thermalization and localization \cite{Bardarson2012, Yang2017}.

Coming to experimental implementations, platforms such as ultracold atoms in optical lattices \cite{Morsch2006, Bloch2012} and superconducting circuits \cite{Buluta2009, Devoret2013} are ideal for experimentally studying the transition by comparing the distribution sampled by the algorithm and the experimental distribution.
The ingredients required, which have been realized in several groups \cite{Bakr2009, Bakr2010, Sherson2010, Mazza2015, Miranda2015, Yamamoto2017}, are: 1.\ Preparation of the initial state \cite{Weitenberg2011, Wang2015} of the type shown in Fig.\ \ref{fig_lattice} (see note \footnote{All of our arguments go through even if the initial state is not regular as in Fig.\ \ref{fig_lattice} but still has the property that the initially occupied sites are separated by a minimum distance of $\Theta(L)$ from each other.}).
2.\ Evolution under a Hamiltonian with either arbitrary time-dependent nearest-neighbor hopping strength or fixed nearest-neighbor hopping strength $J_{ij}(t) = 1$ together with arbitrary time-dependent on-site potential $J_{ii}(t)$ \cite{Weitenberg2011, Fukuhara2013}, and
3.\ Single-site resolved measurement of occupation number of the sites \cite{Bakr2010, Sherson2010}.
Cold atoms in quantum gas microscopes can be controlled at the single site level, enabling all three ingredients above.
To maintain integrability, we can turn off the Hubbard interaction for the bosonic atoms by tuning to a Feshbach resonance \cite{Paul2016, Xu2003, Herbig2003a, Durr2004}.

New architectures like optical tweezers \cite{Kaufman2014} are also promising since they allow for deterministic creation of desired initial states and feature tunable interactions in time \cite{Endres2016, Bernien2017}.
Similarly, superconducting circuits have been proposed for quantum simulation of quantum walks \cite{Flurin2017} and the Bose-Hubbard model \cite{Hacohen-Gourgy2015, Deng2016, Roushan2017}.
The gmon qubit architecture, which was used in Ref.\ \cite{Roushan2017}, naturally allows for time-dependent variation of coupling strengths \cite{Chen2014}, and raises the prospect of an experimental study of the transition.

\begin{acknowledgments}
\textit{Acknowledgments.---}
We are grateful to Mohammad Hafezi, Victor Galitski, Zhe-Xuan Gong, Sergey Syzranov, Ignacio Cirac, Emmanuel Abbe and Elizabeth Crosson for discussions.
A.D., M.C.T., and A.V.G.\ acknowledge funding from ARL CDQI, ARO MURI, NSF QIS, ARO, NSF PFC at JQI, and AFOSR.
B.F.\  is  funded  in  part  by  AFOSR  YIP No.\ FA9550-18-1-0148  as  well  as  ARO  Grants No.\ W911NF-12-1-0541 and No.\ W911NF-17-1-0025, and NSF Grant No.\ CCF-1410022.

Recently, we learned of a related work \cite{Muraleedharan2018}, which shows an analogous result for logarithmic time without constraints on the initial state.
\end{acknowledgments}

\begin{widetext}
\renewcommand{\thesection}{S\arabic{section}} 
\renewcommand{\theequation}{S\arabic{equation}}
\renewcommand{\thefigure}{S\arabic{figure}}
\setcounter{equation}{0}
\setcounter{figure}{0}
\section{Supplemental material}
In this supplemental material, we give expressions for the output probabilities in the distribution $\mathcal{D}_U$ in a boson sampling experiment.
We then explicitly present the algorithm and derive the expression for $\mathcal{D}_{DP}$.
We then derive an upper bound to the variation distance between them, proving Lemma 3 of the main text.

\textit{Expression for output probabilities.---} In this section, we describe the standard boson sampling set-up and derive an expression for the output probabilities of a boson sampling experiment that define the distribution $\mathcal{D}_U$.
First, let us represent the input and output states pictorially and develop some notation.

\begin{figure}[h]
 \includegraphics[width=0.5\linewidth]{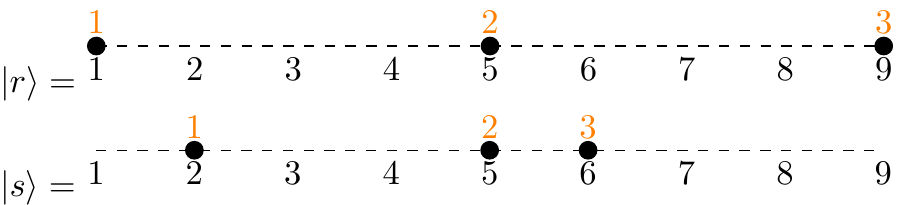}
 \caption{A representation of input and output basis states in 1-D.} \label{fig_SM}
\end{figure}

In Fig.\ \ref{fig_SM}, the top line denotes the input state $\ket{r}$ and the bottom line the output $\ket{s}$.
Each filled circle denotes a boson occupying the corresponding lattice site, which is labeled below the circles.
The numbers marked in orange above each boson label the bosons from left to right (more generally, this is in a nondecreasing order of the site index).
We will call this the boson index.

A given configuration (basis state) is completely specified by specifying the boson number in each site, such as $r = (1,0,0,0,1,0,0,0,1)$ and $s = (0,1,0,0,1,1,0,0,0)$ in the above.
It can also be specified by listing the site index for every boson index, i.e.\ the occupied sites.
Thus the input state can be represented as $\mathsf{in} = (1,5,9)$, the output state as $\mathsf{out} = (2, 5, 6)$.

All $n!$ permutations of the boson indices represent valid paths that the bosons can take to the output state, and correspond to the $n!$ terms in the permanent of the matrix.
In cases where there are two or more bosons in a particular site at the input or output, there are $\frac{n!}{r! s!}$ paths (and terms in the amplitude).
Here, $r! := r_1! r_2! \ldots r_m!$ and similarly $s!$.
By taking repeated rows and columns of $R$, this has the effect of still giving $n!$ terms in total, which we identify with the $n!$ permutations in the boson indices.
The expression for the probability of an outcome $s$ is (here, $b_i := a_i(t)$):

\begin{gather}
  \Pr_{\mathcal{D}_U}[s] = \frac{1}{r_1!r_2!\ldots r_m! s_1! s_2! \ldots s_m!} 
|\bra{\mathrm{vac}} 
  b_1^{s_1} b_2^{s_2} \ldots b_m^{s_m} a_1^{\dagger^{r_1}} a_2^{\dagger^{r_2}} \ldots 
a_m^{\dagger^{r_m}}\ket{\mathrm{vac}} |^2 \\
= \frac{1}{r! s!} |\bra{\mathrm{vac}} (\hat{U}^\dagger a_1^{s_1} \hat{U}) 
(\hat{U}^\dagger a_2^{s_2} \hat{U}) \ldots (\hat{U}^\dagger a_m^{s_m} \hat{U}) 
a_1^{\dagger^{r_1}} a_2^{\dagger^{r_2}} \ldots a_m^{\dagger^{r_m}} \ket{\mathrm{vac}}|^2 \\
= \frac{1}{r! s!}  |\bra{\mathrm{vac}} \left(\sum_{k_1=1}^{m} R^\dagger_{1k_1} a_{k_1} 
\right)^{s_1} \ldots \left(\sum_{k_m=1}^{m} R^\dagger_{mk_m} a_{k_m} \right)^{s_m} 
a_1^{\dagger^{r_1}} a_2^{\dagger^{r_2}} \ldots a_m^{\dagger^{r_m}} \ket{\mathrm{vac}}|^2,
\end{gather}
where $R_{ij}$ describes the action of $\hat{U}$ on the annihilation operators at a site: 
$b_i = a_i(t) = \sum_k R^\dagger_{ik}(t)a_k(0)$.
Now define the matrix $A^\dagger$ to be the one obtained by taking $s_i$ copies of the $i$'th row and $r_j$ copies of the $j$'th column of $R^\dagger$.
For concreteness, this can be done by first considering the rows and repeating a row $i$ of $R^\dagger$ whenever $s_i > 1$, or not picking it if $s_i=0$, to convert it into an $n\times m$ matrix.
We can then do the same with columns to convert it into an $n \times n$ matrix.
However, the order ultimately does not matter since the quantity that emerges, the permanent, is symmetric under exchange of rows or columns.
We have
\begin{gather}
 \Pr_{\mathcal{D}_U}[s] = \frac{1}{r! s!} \left| \sum_\sigma \prod_i 
R^\dagger_{\mathsf{out}_{\sigma(i)},\mathsf{in}_i} \right|^2 \label{eq_probR} \\ 
=  \frac{1}{r! s!}  \left| \sum_\sigma \prod_i A^\dagger_{\sigma(i),i} \right|^2, 
\label{eq_probA}
\end{gather}
where the sum is over all permutations $\sigma$.
This finally gives us 
\begin{gather}
 \Pr_{\mathcal{D}_U}[s] = \frac{1}{r!s!} 
\left|\mathrm{Per}(A^\dagger) \right|^2 = \frac{1}{r!s!}|\mathrm{Per}(A)|^2,
\end{gather}
where $\mathrm{Per}(A) $ is the \textit{permanent} of $A$.

\textit{Algorithm.---}
The sampling algorithm is given below.
It is easy to see that it implements one step of a Markov process of $n$ distinguishable bosons walking on a lattice.
\SetNlSkip{0.8em}
\SetInd{0.5em}{1em}

\begin{algorithm}[H]
 \caption{Sampling algorithm} \label{alg_main}
 \SetInd{1em}{1em}
\DontPrintSemicolon
 \KwIn{Unitary $R(t)$, tolerance $\epsilon$}
 \KwOut{Sample $s$ drawn from $\mathcal{D}_{DP}$, a distribution that is close to $\mathcal{D}_U$.}
\ $\mathcal{P}_{kl} = |R(t)|_{kl}^2$ \label{alg_line_Markovmat} \;
 \For{$i \in \{1,2,\ldots n\}$,}{
Select site $l$ from the distribution $\mathcal{P}_{\mathsf{in}_i,l}$ for the boson at $\mathsf{in_i}$ to hop to.\label{alg_line_distb}\;
Increment output boson number of site $l$ by 1: $s_l \rightarrow s_l + 1$ (or equivalently, assign $\mathsf{out}_i = l$)\;
}
\KwRet configuration $s$ (or $\mathsf{out}$), a sample from $\mathcal{D}_{DP}$.

\end{algorithm}

Note that $\mathcal{P}$ from line \ref{alg_line_Markovmat} is a doubly stochastic matrix describing the classical Markov process.
To see that the runtime is polynomial in $n$, note that the loop is over $n$ boson indices.
Line \ref{alg_line_distb} takes time $O(m \log m) = \tilde{O}(n^\beta)$, giving a total runtime of $\tilde{O}(nm) = \tilde{O}(n^{1+\beta})$.
The notation $\tilde{O}$ suppresses factors of $\log n$.

\textit{Bound on variation distance.---} Here we derive a bound on the variation distance  $\|\mathcal{D}_U - \mathcal{D}_{DP}\| = \frac{1}{2} \sum_s |\Pr_{\mathcal{D}_U}(s) - \Pr_{\mathcal{D}_{DP}}(s)|$.
Rewriting the actual probability in terms of the amplitudes, we have 
\begin{gather}
 \Pr_{\mathcal{D}_U}(s) = |\phi|^2, \ \ \mathrm{with} \\
 \phi^* = \frac{1}{\sqrt{r!s!}}\sum_\sigma 
R_{\mathsf{in}_1,\mathsf{out}_{\sigma(1)}}R_{\mathsf{in}_2, \mathsf{out}_{\sigma(2)}} \ldots 
R_{\mathsf{in}_n, \mathsf{out}_{\sigma(n)}} \\
= \frac{1}{\sqrt{s!}} \sum_\sigma A_{1,\sigma(1)} A_{2,\sigma(2)} \ldots A_{n,\sigma(n)} 
\label{eq_amplitude},
\end{gather}
where $A$ is the $n \times n$ matrix formed by taking the appropriate number of copies of each row and column of $R_{m \times m}$.
We have set $r! = 1$ since our input state has bosons in distinct sites.
Continuing,
\begin{gather}
 \Pr_{\mathcal{D}_U}(s) = \frac{1}{s!} \sum_\sigma |R_{\mathsf{in}_1,\mathsf{out}_{\sigma(1)}}|^2 
|R_{\mathsf{in}_2, 
\mathsf{out}_{\sigma(2)}}|^2 \ldots |R_{\mathsf{in}_n, \mathsf{out}_{\sigma(n)}}|^2 \ + \nonumber 
\\ 
 \frac{1}{s!} \sum_{\sigma \neq \tau} 
R_{\mathsf{in}_1,\mathsf{out}_{\sigma(1)}}R_{\mathsf{in}_2, \mathsf{out}_{\sigma(2)}} \ldots 
R_{\mathsf{in}_n, \mathsf{out}_{\sigma(n)}} 
(R_{\mathsf{in}_1,\mathsf{out}_{\tau(1)}}R_{\mathsf{in}_2, \mathsf{out}_{\tau(2)}} \ldots 
R_{\mathsf{in}_n, \mathsf{out}_{\tau(n)}})^*.
\label{eq_prob_actual}
 \end{gather}

The probability distribution $\mathcal{D}_{DP}$ that the algorithm samples from is given by the first line of Eq.\ (\ref{eq_prob_actual}):
\begin{gather}
 \Pr_{\mathcal{D}_{DP}}(s) =  \frac{1}{s!} \sum_\sigma 
\mathcal{P}_{\mathsf{in}_1,\mathsf{out}_{\sigma(1)}} 
\mathcal{P}_{\mathsf{in}_2,\mathsf{out}_{\sigma(2)}} \ldots 
\mathcal{P}_{\mathsf{in}_n,\mathsf{out}_{\sigma(n)}}, 
\end{gather}
where the sum is over all the $n!$ ways of assigning the $n$ input states to the $n$ output states.
As before, the $s!$ is to account for overcounting when two distinct permutations in the boson index refer to the same site index in the output state.

The expression for the probability is proportional to the permanent of the matrix with the positive entries $\mathcal{P}_{\mathsf{in}_i,\mathsf{out}_j}$, and can hence be efficiently approximated \cite{Jerrum2004}.
Note that the algorithm does not explicitly calculate this probability but only samples from the distribution.
We can now prove Lemma \ref{lem_main} of the main text.

\begin{proof}[Proof of Lemma \ref{lem_main}]
The variation distance is given by 
\begin{align}
 \varepsilon &= \sum_s \frac{1}{2s!} \left| \sum_{\sigma \neq \tau} R_{\mathsf{in}_1,\mathsf{out}_{\sigma(1)}} \ldots R_{\mathsf{in}_n, \mathsf{out}_{\sigma(n)}} (R_{\mathsf{in}_1,\mathsf{out}_{\tau(1)}} \ldots R_{\mathsf{in}_n, \mathsf{out}_{\tau(n)}})^* \right|\\
& \leq \sum_s \frac{1}{2} \sum_{\sigma \neq \tau} |R_{\mathsf{in}_1,\mathsf{out}_{\sigma(1)}} \ldots R_{\mathsf{in}_n,\mathsf{out}_{\sigma(n)}}| |R_{\mathsf{out}_{\tau(1)}, \mathsf{in}_1} \ldots R_{\mathsf{out}_{\tau(n)}, \mathsf{in}_n}| \\
& = \sum_s \frac{1}{2} \sum_{\sigma, \rho} |R_{\mathsf{in}_1,\mathsf{out}_{\sigma(1)}} R_{\mathsf{out}_{\sigma(1)},\mathsf{in}_{\rho(1)}}| \ldots |R_{\mathsf{in}_n,\mathsf{out}_{\sigma(n)}} R_{\mathsf{out}_{\sigma(n)}, \mathsf{in}_{\rho(n)}}|, \label{eq_vardistanceintermediate} 
\end{align}

where $\rho = \tau^{-1}\circ \sigma \neq \mathsf{Id}$, the identity permutation.
The last equality comes from rearranging the terms in $|R_{\mathsf{out}_{\tau(1)}, \mathsf{in}_1} \ldots R_{\mathsf{out}_{\tau(n)}, \mathsf{in}_n}|$ so that the terms involving $R_{\mathsf{in}_i,\mathsf{out}_{\sigma(i)}}$ and $R_{\mathsf{out}_{\sigma(i)},j}$ (for some $j$) are together:
\begin{align}
 \sum_\sigma \sum_\tau \prod_i |R_{\mathsf{out}_{\tau(i)},\mathsf{in}_i}| = \sum_\sigma \sum_\tau \prod_i |R_{\mathsf{out}_{i},\mathsf{in}_{\tau^{-1}(i)}}| \xrightarrow[i\rightarrow \sigma(i)]{\text{rearrange}}  \sum_\sigma \sum_\tau \prod_i |R_{\mathsf{out}_{\sigma(i)},\mathsf{in}_{\tau^{-1}(\sigma(i))}}|
\end{align}

Summing over all outcomes $s$ (or configurations $\mathsf{out}$), Eq.\ (\ref{eq_vardistanceintermediate}) is equivalent to 
\begin{align}
 \varepsilon \leq \frac{1}{2} \sum_j \sum_{\rho\neq \mathsf{Id}} \prod_{i} |R_{\mathsf{in}_i,j_i}||R_{j_i, \mathsf{in}_{\rho(i)}}|, \label{eq_basepointMinh}
\end{align}
where the sum $j$ is over ordered tuples $(j_1,\ldots j_n)$, representing the intermediate lattice sites that the bosons in positions $(\mathsf{in}_1,\ldots \mathsf{in}_n)$ jump to, before jumping back to positions $(\mathsf{in}_{\rho(1)},\ldots \mathsf{in}_{\rho(n)})$.

We can proceed to break the sum in Eq.\ (\ref{eq_basepointMinh}) based on the number of fixed points of the permutation $\rho$, that is, the number of indices $i$ such that $\rho(i) = i$.
We bound these quantities separately as follows:
\begin{align}
 & \sum_{\substack{j_i, \rho\\ \rho(i) \neq i }} |R_{\mathsf{in}_i,j_i}||R_{j_i,\mathsf{in}_{\rho(i)}}| = C_i \leq c\ \forall\ i \ \ \mathrm{and} \nonumber \\
& \sum_{\substack{j_i, \rho\\ 
 \rho(i) = i}} |R_{\mathsf{in}_i,j_i}||R_{j_i,\mathsf{in}_{\rho(i)}}| = \sum_{j_i} |R_{\mathsf{in}_i,j_i}|^2 = D_i = 1\ \forall\ i. \label{eq_bound_c}
\end{align}

The variation distance is therefore bounded above:
\begin{align}
 \varepsilon \leq \frac{1}{2} \sum \prod_{i\in \mathcal{I}_C} C_{i} \prod_{k\in \mathcal{I}_D} D_k, \label{eq_R_prime}
\end{align}
where the sum is over subsets $\mathcal{I}_C$ of the indices representing the input state, $\mathsf{in}$.
$\mathcal{I}_D$ is the complement of $\mathcal{I}_C$ and $|\mathcal{I}_D|$ is the number of fixed points.
Suppose we find an upper bound $c$ for $C_i$ in Eq.\ (\ref{eq_bound_c}), we then have
\begin{align}
 \varepsilon \leq \frac{1}{2} \sum_{l=|\mathcal{I}_C|=2}^{n} \binom{n}{l}c^l = (c+1)^n - nc - 1. \label{eq_bound_binomial}
\end{align}

In Lemma \ref{lem_boundc}, we show that $c = \eta L^{d-1} \mathrm{e}^{(vt-L)/\xi}$ for some constant $\eta$.
Continuing from Eq.\ (\ref{eq_bound_binomial}),
\begin{align}
\varepsilon & \leq \frac{1}{2} [(c+1)^n -nc - 1] \\
&= \binom{n}{2} (1+h)^{n-2}  c^2 \ \mathrm{for\ some\ } h \in [0,c] \ (\mathrm{by\ Taylor's\ theorem}) \\
\varepsilon & \leq \exp \left[2\log n + (n-2) \log (1 + c) + 2 \log c \right]
\end{align}
Now, plugging in the value of $c$ and assuming that $vt \leq 0.9L$ and $\beta > 1$, we get
\begin{align}
 \varepsilon & \leq O \left(\exp \left[(n-2)\times \eta L^{d-1} \mathrm{e}^{(vt-L)/\xi} +  2 \frac{vt-L}{\xi} + 2(d-1)\log L \right] \right) \\
& \leq O\left(\exp \left[2\frac{vt-L}{\xi} + 2(d-1)\log L \right]\right).
\end{align}
In the first line, we use the inequality $\log(1+x) \leq x$ and in the second, $\exp \left[(n-2)\times \eta L^{d-1} \mathrm{e}^{(vt-L)/\xi}\right] = O(1)$ since $vt< L$ and $|vt - L| = \Omega(n^{\beta - 1})$.
This completes the proof of Lemma \ref{lem_main}.
\end{proof}

\begin{lemma}\label{lem_boundc}
For all constant dimensions $d$, $c = \eta L^{d-1} \mathrm{e}^{(vt-L)/\xi}$.
\end{lemma}
\begin{proof}
Recall that
\begin{align}
C_i = \sum_{j_i}\sum_{\substack{\rho(i) \neq i}} 
\left|R_{\mathsf{in}_i,j_i}\right|\left|R_{j_i,\mathsf{in}_{\rho(i)}}\right|.
\end{align}
Since we are looking for a bound that applies for all $i$, let us, for convenience, make the following changes in notation: $\mathsf{in}_i \rightarrow i, j_i \rightarrow j, \mathsf{in}_{\rho(i)} \rightarrow k$, denoting a boson starting at position $i$, jumping to $j$ and then to $k$, where $i$ and $k \neq i$ are site indices belonging to $\mathsf{in}$.
We split the sum in $C_i$ based on the distance between $i$ and $j$, $\ell_{ij} =: \ell$.
\begin{align}
 C_i = \sum_{\substack{j\\ \ell \leq L}} \sum_{\substack{k \in \mathsf{in}\\k \neq i}} \left|R_{i,j}\right| \left|R_{j, 
k}\right| + \sum_{\substack{j\\ \ell > L}} \sum_{\substack{k \in \mathsf{in}\\k \neq i}}\left|R_{i,j}\right| \left|R_{j, k}\right|. \label{eq_S_Ci_distances}
\end{align}
Consider the first term:
\begin{align}
\sum_{\substack{j\\\ell \leq L}} 
\left|R_{i,j}\right|\sum_{\substack{k\in \mathsf{in}\\k\neq i}}\left|R_{j,k}\right| & \leq \sum_{\substack{j\\ \ell \leq L}} \left|R_{i,j} \right| \sum_{\substack{k\in \mathsf{in} \\k\neq i}}\mathrm{e}^{(vt-\ell_{kj})/\xi} \label{EQ_F1_Step2}\\
&\leq \sqrt{\left(\sum_{\substack{j\\ \ell \leq L}} 1^2\right)\left(\sum_{\substack{j\\ \ell \leq L}} \left|R_{i,j}\right|^2\right)} \mathrm{e}^{vt/\xi} \sum_{k:\|\vec{x}_k\|\geq 1} \mathrm{e}^{(-2L\|\vec{x}_k \| + L)/\xi} \label{EQ_F1_Step3}\\
&\leq a L^{d/2} e^{(vt +L)/\xi}\sum_{k:||\vec x_k||\geq 1} e^{-2L||\vec x_k||/\xi} \label{EQ_F1_Step4}\\
&\leq ab e^{(vt- L)/\xi} L^{d/2}.
\end{align}

Here in the first line, we have used the Lieb-Robinson bound Eq.\ (\ref{eq_LRbound}).
In the second line, we use the Cauchy-Schwarz inequality and the fact that $\ell = \ell_{ij} \leq L$.
In the second line, $\vec{x}_k $ is the position vector of site $k$ relative to site $i$, re-scaled by $2L$.
Therefore the sum over $\vec{x}_k $ is over all vectors with integer coordinates.
In the last line, we use Lemma \ref{lem_dimension_better}, to be proven below.
$a$ and $b$ are constants independent of $n$ that depend on the dimension $d$ and the length scale $\xi$.

Now, in the second term for $C_i$ in Eq.\ (\ref{eq_S_Ci_distances}), the intermediate site $j$ is not necessarily close to $i$.
Therefore, there are terms where $j$ is close to $k \neq i$ and one has to treat these terms carefully.
For these terms, we use the trivial Lieb-Robinson bound of 1 in Eq.\ (\ref{eq_LRbound}) rather than $\exp \left( \frac{vt - \ell_{jk}}{\xi}\right) > 1$.

\begin{align}
 \sum_{\substack{j\\ \ell > L}} \left|R_{i,j}\right|\sum_{\substack{k\in \mathsf{in} \\ k \neq i}} \left| R_{j,k}\right| & \leq \sum_{\substack{j \\ \ell > L}} \left|R_{i,j} \right| \sum_{\substack{k \in \mathsf{in} \\ k \neq i}} \mathrm{min}(1, \mathrm{e}^{(vt-\ell_{k,j})/\xi})\\
&\leq \sum_{\substack{j\\\ell > L}} \left|R_{i,j}\right|\left(1 + \mathrm{e}^{vt/\xi} \sum_{k:||\vec x_{k}||\geq 1} \mathrm{e}^{(L-2L\|\vec x_{k}\|)/\xi} \right)\\
& \leq \left(1 + b \mathrm{e}^{(vt- L)/\xi} \right) \sum_{\substack{j\\ \ell > L}} \left|R_{i,j}\right| \\
&\leq \left(1 + b \mathrm{e}^{(vt- L)/\xi}\right) \mathrm{e}^{vt/\xi} \sum_{\substack{j\\ \ell> L}} \mathrm{e}^{-\ell/\xi} \\
&\leq \left(1 + b \mathrm{e}^{(vt- L)/\xi} \right) \tilde{b} L^{d-1} \mathrm{e}^{(vt-L)/\xi}.
\end{align}
In the second line, we use 1 as a Lieb-Robinson bound for $|R_{j,k}|$ when $k = k^*$, the site belonging to $\mathsf{in}$ that is closest to $j$.
All other $k$'s have distances from $j$ bounded below by $2L\|\vec{x}_k\| - L$, where $\vec{x}_k$ is now the re-scaled position vector of a site $k$ with the origin at $k^*$.
We apply Lemma \ref{lem_dimension_better} in the third line and Lemma \ref{lem_dimension} in the fifth.
Collecting everything, we have
\begin{align}
 C_i  & \leq ab \mathrm{e}^{(vt-L)/\xi} L^{d/2} + \left(1 + b \mathrm{e}^{(vt- L)/\xi} \right) \tilde{b} L^{d-1} \mathrm{e}^{(vt-L)/\xi} \\
 & \leq \mathrm{e}^{(vt- L)/\xi} \eta L^{d-1} \ \mathrm{for\ large\ enough} \ L.
\end{align}
\end{proof}

\begin{lemma}[$d$-dimensional sum]\label{lem_dimension}
The sum $\sum_{\|\vec{x}\| \geq L} \mathrm{e}^{-\|\vec{x}\|/\xi}$ over points $\vec{x}$ with integer coordinates is upper bounded by $a_d \mathrm{e}^{- L/\xi} \left(\xi L^{d-1}\right)$ for large enough $\frac{L}{\xi}$ for some dimension-dependent constant $a_d$.
\end{lemma}
\begin{proof}
We can view the sum over a lattice of vectors with integer coordinates as a Riemann sum and bound the corresponding $d$-dimensional integral.
Consider the quantity
\begin{align}
g = \int_{\|\vec{x}\| \geq L} \mathrm{e}^{-\|\vec{x}\|/\xi} \mathrm{d}^d \vec{x} = \frac{2\pi^{d/2}}{\Gamma(\frac{d}{2})} \xi^d \Gamma \left(d, \frac{L}{\xi}\right), \label{eq_S_Riemann}
\end{align}
where $\Gamma(d,x) = \int_x^\infty y^{d-1} \mathrm{e}^{-y} \mathrm{d}y$ is the incomplete $\Gamma$ function.
We can lower bound the integral by the Riemann sum $\sum_\Delta V_\Delta \mathrm{e}^{-\|\vec{y}\|/\xi}$, where the sum is over cells $\Delta$ with volume $V_\Delta$ centered at lattice points $\vec{x}$.
$\vec{y}$ is the point in the cell $\Delta$ with the highest norm $\|\vec{y}\|$.
Further, the point with the highest norm is not too distant from the one at the center: $\|\vec{y}\| \leq \|\vec{x}\| + \frac{\sqrt{d}}{2}$.
Therefore, we have 
\begin{align}
 f := \sum_{\|\vec{x}\| \geq L} \mathrm{e}^{\|\vec{x}\|/\xi} \leq g \times \mathrm{e}^{\sqrt{d}/(2\xi)}. \label{eq_S_fbound}
\end{align}
We now need an upper bound on the incomplete $\Gamma$ function $\Gamma(d,x)$ for large $x$ \cite{Olver}:
\begin{align}
 \Gamma(d,x) \rightarrow x^{d-1} \mathrm{e}^{-x} \left(1 + O\left(\frac{1}{x}\right) \right) \mathrm{as}\ x \rightarrow \infty.
\end{align}
Combining Eq.\ (\ref{eq_S_Riemann}) and Eq.\ (\ref{eq_S_fbound}), we get:
\begin{align}
 f \leq O\left(\xi L^{d-1} \exp \left[ - \frac{L}{\xi}\right] \right).
\end{align}
\end{proof}
Using a similar method, we also get bounds on a related sum.
\begin{figure}
\includegraphics[width=0.4\linewidth]{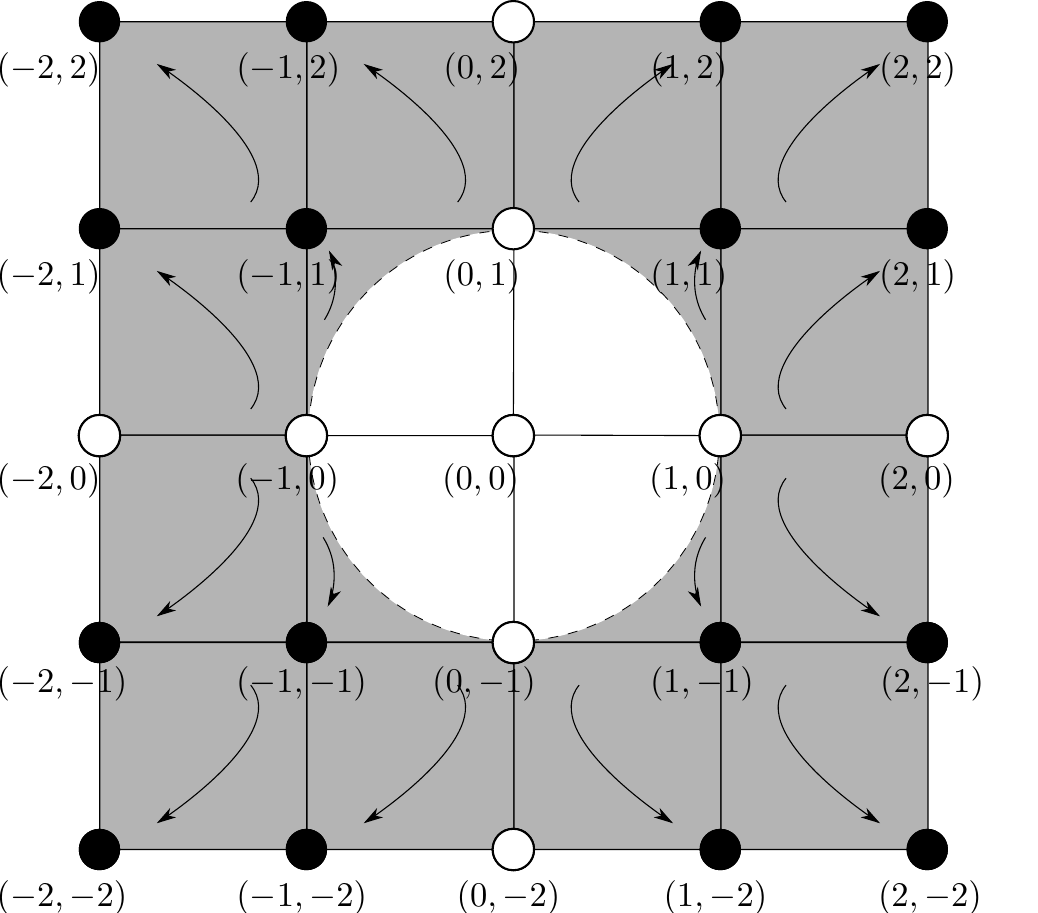}
\caption{Part of the lattice of vectors with integer coordinates.
The black dots are the points in the cell with the maximum norm $\|\vec{y}\|$ and the exponential is evaluated at these points.
The white ones do not enter the Riemann sum and are related to $f_{d-1}$, the corresponding quantity in one lower dimension.
The arrows show which point in the cell is picked to lower bound the Riemann sum.} \label{fig_SM_tiling}
\end{figure}
\begin{lemma}\label{lem_dimension_better}
 For $\vec{x} \in \mathbb{Z}^d$, $f_d := \sum_{\|\vec{x}\| \geq 1} \mathrm{e}^{-2L\|\vec{x}\|/\xi} \leq b_d \exp[{-\frac{2L}{\xi }}]$ for some dimension-dependent constant $b_d$.
\end{lemma}
\begin{proof}
We prove the statement by induction on the dimension $d$.
For $d = 1$, the statement is seen to be true since the sum evaluates exactly:
\begin{align}
 f_1 = \sum_{|x|\geq 1} \mathrm{e}^{-2L |x|/\xi} = 2 \sum_{x=1}^\infty \mathrm{e}^{-2L x/\xi} = \frac{2 \mathrm{e}^{-2L/\xi}}{1- \mathrm{e}^{-2L/\xi}} \leq 2.1 \times \mathrm{e}^{-2L/\xi}.
\end{align}
For the inductive step, consider the integral $g(d) = \int_{\|\vec{y}\| \geq 1}^{\infty} \mathrm{e}^{-2L\|\vec{y}\|/ \xi} \mathrm{d}^d\vec{y}$.
This is lower bounded by the Riemann sum represented in Fig.\ \ref{fig_SM_tiling}.
The white dots represent vectors with at least one zero coordinate and do not enter the Riemann sum according to this way of dividing the region of integration into cells.
In the following, the set of points with at least one zero coordinate is denoted $\O_{cc}$.
We have:
\begin{align}
  g(d) = \int_{\|\vec{y}\| \geq 1}^{\infty}  \mathrm{e}^{-2L\|\vec{y}\|/\xi} \mathrm{d}^d\vec{y} \geq \sum_{ \vec{x} \notin  \O_{cc}} \mathrm{e}^{-2L\|\vec{x}\|/\xi} \Delta_{\vec{x}},
\end{align}
where $\Delta_{\vec{x}}$ is the volume of the cell associated with the lattice vector $\vec{x}$.
In Fig.\ \ref{fig_SM_tiling}, the volume of most cells (whose center is at distance 1.5 or beyond from the origin) is 1.
The cells near the origin have some volume $\alpha_d < 1$ that depends on the dimension.
Lower bounding all volumes $\Delta_{\vec{x}}$ by $\alpha_d$,
\begin{align}
 \sum_{ \vec{x} \notin  \O_{cc}} \mathrm{e}^{-2L\|\vec{x}\|/\xi} & < \frac{g(d)}{\alpha_d} \\
 & = \frac{2\pi^{d/2}}{\alpha_d \Gamma(\frac{d}{2})} \left(\frac{\xi}{2L}\right)^d \Gamma \left(d, \frac{2L}{\xi}\right). \label{eq_S_blackcells}
\end{align}
Now it remains to upper bound the contribution from summing over the points $\O_{cc}$.
Notice that the sum over these points is upper bounded by the sum over $d$ hyperplanes of dimenson $d-1$.
From the inductive hypothesis,
\begin{align}
 \sum_{ \vec{x} \in  \O_{cc}} \mathrm{e}^{-2L\|\vec{x}\|/\xi} \leq df_{d-1} \label{eq_S_whitecells}
\end{align}
since there are $d$ hyperplanes of dimension $d-1$.
Adding Eqs.\ (\ref{eq_S_blackcells}) and (\ref{eq_S_whitecells}), we get
\begin{align}
 \sum_{\vec{x}} \mathrm{e}^{-2L\|\vec{x}\|/\xi} = f_d & < df_{d-1} + \frac{2\pi^{d/2}}{\alpha_d \Gamma(\frac{d}{2})} \left(\frac{\xi}{2L}\right)^d \Gamma \left(d, \frac{2L}{\xi}\right). \\
 & \leq db_{d-1} \exp \left[- \frac{2L}{\xi} \right] + \frac{2\pi^{d/2}}{\alpha_d \Gamma(\frac{d}{2})} \left(\frac{\xi}{2L}\right) \exp \left[ -\frac{2L}{\xi}\right] \\
 f_d & < b_d \exp \left[ -\frac{2L}{\xi}\right],
\end{align}
proving the lemma. In the second line we have expanded the incomplete $\Gamma$ function for large $L/\xi$.
\end{proof}
\end{widetext}

\end{document}